\documentclass[12pt,conference,compsoc,onecolumn]{IEEEtran}

\usepackage{preamble}

\title{Quantization using Compressive Sensing}

\author{\IEEEauthorblockN{Rajiv Soundararajan and Sriram Vishwanath}
\IEEEauthorblockA{Department of Electrical and Computer Engineering\\ The University of Texas at Austin\\
1 University Station C0803, Austin, TX 78712, USA\\
Email: rajiv.s@mail.utexas.edu and sriram@ece.utexas.edu}}

\begin{document}

\maketitle

\begin{abstract}
The problem of compressing a real-valued sparse source using compressive sensing techniques is studied. The rate distortion optimality of a coding scheme in which compressively sensed signals are quantized and then reconstructed is established when the reconstruction is also required to be sparse. The result holds in general when the distortion constraint is on the expected $p$-norm of error between the source and the reconstruction. A new restricted isometry like property is introduced for this purpose and the existence of matrices that satisfy this property is shown. 
\end{abstract}

\section{Introduction}

In recent years, there has been an explosion in the number of applications to which compressed (or compressive) sensing has been applied \cite{Duarte2008,Vikalo2008}. From image and video capture to microarrays and other applications represents just a sub-spectrum of its possible uses. An important application that seems particularly promising in terms of being translated into practical circuit designs is {\em quantization}. If the original signal to be compressed is $n$ dimensional but is $k$-sparse (i.e., has at most $k$ non-zero entries), and $k << n$, then there is a significant benefit in using a compressive sensing framework for quantization. Indeed, compressive sensing in itself represents a nearly-lossless linear transformation on the original source, and thus, compressive sensing is a ``good lossless" compression mechanism for sparse vectors, reducing the length of the representation of the original source from an $n$ dimensional vector to $m = \Theta(k \log \frac{n}{k})$ dimensional vector, which is an order-wise optimal lossless compression of the source. In particular, ``sampling" matrices $\Phi$ of dimension $m \times n$ have been shown to exist where
\begin{equation}
\label{eqn:cs}
y^m = \Phi x^n,
\end{equation}
such that the original source $x^n$ can be recovered losslessly with high probability. This is not surprising from a compression perspective, as optimal linear compressors are known to exist for lossless compression.

Our goal in this paper is to investigate if compressive sensing is good for lossy compression of continuous-valued sources. The answer to this question is not immediately obvious, as typical lossy compression  algorithms involve non-linear transformations between the source and its compressed equivalent in the encoding step. There are two ways in which compressive sensing can be combined with quantization. The first is where the number of samples $m$ in (\ref{eqn:cs}) is reduced from $\Theta(k \log \frac{n}{k})$ to an orderwise smaller quantity. The resulting lossy-compressed vector $y$ is then ``inverted" to obtain a compressed version of $x$. This approach is studied in detail for the case when $ k = \Theta(n)$ in \cite{Reeves2010}. The second is to maintain the lossless compression process as prescribed by (\ref{eqn:cs}) and then use a (possibly non-linear) quantizer on $y^m$ to obtain $\hat{y}^m$. Subsequently,  $\hat{y}^m$ is transformed using a suitable algorithm to obtain $\hat{x}^n$, a lossy-compressed version of $x^n$. This second approach has been studied recently in \cite{Baraniuk2009}, and it forms the main quantization scheme studied in this paper.

\begin{figure}\label{fig:coding}
\centering
\includegraphics[scale=0.5]{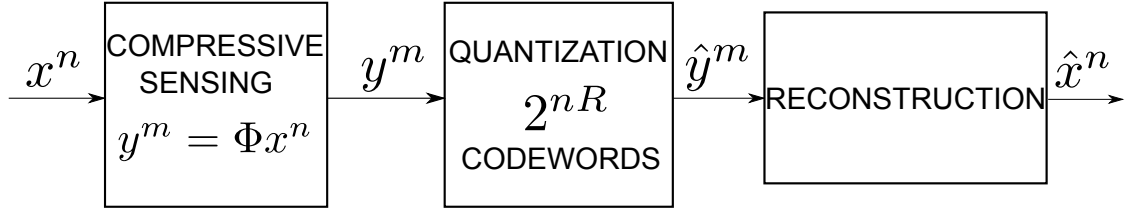} 
\caption{Compressive sensing followed by quantization}
\end{figure}

The chain comprising of sampling using compressive sensing, then quantization and finally, reconstruction of the compressed vector is depicted in Fig. \ref{fig:coding}. Associated with this framework are two notions of rate, the sampling rate and the compression rate. While the sampling rate is the rate at which the quantizer needs to the sample the incoming signal, the compression rate is the number of bits per symbol needed to represent the sampled signal within a fidelity criterion. This chain is particularly useful in developing A/D converters for sparse sources - it reduces the sampling rate at which they must operate thus making their design simpler and the quantization operation more effective. What we desire to know in this paper is if this  quantization mechanism, besides being practically efficient, is indeed optimal. In other words, if the source were to be directly quantized (using the best quantizer available), would it suffer a lower distortion than being first filtered in accordance with (\ref{eqn:cs}) and then compressed? Observe that the compression rate at which the quantizer in Fig. 1 operates is higher than the optimal quantizer so that the number of codeword indices (or the cardinality of the reconstruction alphabet) is kept equal. Mathematically, the compression rate of the quantizer in Fig.1 is $\log(2^{nR})/m$, while the optimal quantizer operates at a compression rate of $R$. The result of this paper may also be interpreted as trading off compression rate for sampling rate while still achieving the same distortion performance as the optimal quantizer. While there is prior literature in studying the performance of specific designs of Fig. 1 \cite{Goyal2009,Dai2010,Boufounos2010}, we prove a conclusive result on when the framework is optimal. 

In this paper, our focus is on those quantization applications where we desire the support of $\hat{x}^n$ and $x^n$ to be identical. This is especially important in applications where we desire that the quantization process not introduce ``spurious" signals. In sensing systems and other applications where the signal represents a change in state of the system, it is particularly important that the compression process retain the original (sparse) support of the original. A distortion in the sparsity pattern could lead to false activation resulting in undesirable consequences. The problem also has applications in DNA microarrays for cancer diagnosis, where a wrong sparsity pattern could lead to faulty diagnosis. 

Our main results are as follows:
\begin{enumerate}
\item The coding architecture in Fig. \ref{fig:coding} is distortion-rate optimal when the reconstruction is also required to be sparse.
\item We show this result when the distortion constraint is on any $p$-norm of the error between the source and reconstruction sequence, where $p\geq 1$. 
\item In order to prove such a general result, we study a modified restricted isometry property (RIP) for matrices and show the existence of matrices that satisfy this property. 
\end{enumerate}

The modified RIP introduced in this paper is essential in order to prove the distortion rate optimality for $p$-norm distortion measures. The proof of existence of matrices satisfying the modified RIP uses Hoeffding's inequality. Related to this work is the use of Hoeffding's inequality  to obtain RIP bounds in a recent paper \cite{Nazer2010}. Also related are results on heavy tailed restricted isometries in \cite{Vershynin2008} and RIP using tail bounds in \cite{Baraniuk2008}. While Hoeffding's inequality has been previously used in other contexts to obtain RIP bounds, this paper uses it to prove the modified RIP required for $p$-norm compression. 

The rest of this paper is organized as follows. In the next section, we describe the system model. In Section \ref{sec:mresult}, we state the main results of the paper. We conclude the paper in Section \ref{sec:conc}. The proofs of the results are detailed in the appendix. 


\section{System Model}\label{sec:sysmodel}
Consider the set ${\mathcal X}_k^n$ of all $k$-sparse vectors of length $n$ where each non-zero entry takes any value in ${\mathbb R}$. The goal is to compress the sparse and real-valued $x^n$ to a vector $\hat{x}^n$ within a distortion $D$. Note that in general, the rate distortion optimal quantizer does not ensure the reconstruction, $\hat{x}^n$, is sparse. Since we desire the support of $\hat{X}^n$ and $X^n$ be identical, we let $\hat{X}^n$ belong to a $k$-sparse reconstruction space denoted by $\hat{\mathcal{X}}^n_k$.  Let $T\subset\{1,2,\ldots,n\}$, be the indices such that $X_i\neq 0$ for $i\in T$. Observe that $T$ is a random set on account of $X^n$ being a random vector. Let $X^n(T)$ be the vector with components corresponding to indices in $T$. $X^n(T^c)$ is defined in a similar fashion. 

We begin by defining the distortion rate function of the optimal quantizer. Let $D_1(R)$ be the average distortion achieved by a code operating at rate $R$ for $k$-sparse source vectors distributed according to $X^n\sim p_{X^n}$. Mathematically,
\begin{align*}
D_1(R) = &\inf_{\hat{\mathcal{X}}^n_k}\frac{1}{n}\mathbb{E}\left[\lVert X^n-\hat{X}^n\rVert_p\right] \\
\textrm{subject to }& \lvert\hat{\mathcal{X}}^n_k\rvert \leq 2^{nR} \textrm{ and } X^n(T^c)=\hat{X}^n(T^c).
\end{align*}
We wish to point out out that the equality constraint in the optimization problem limits the reconstruction spaces to those that are $k$-sparse with the same sparsity pattern as the source. 

We now define the optimization problem concerning the quantizer in Fig. 1 followed by the distortion rate function of the compression scheme in Fig. 1. Let $Y^m=\Phi X^n$ and let $\hat{Y}^m=\Phi \hat{X}^n$ denote the quantized version of $Y^m$. Let $\Phi_i$ be the $i$-th row of a matrix $\Phi$ of dimension $m\times n$, $i=1,2,\ldots,m$ and $q\geq 1$ satisfy $\frac{1}{p}+\frac{1}{q}=1$. Define,
\begin{align*}
\Delta_2(R)=&\inf_{\hat{\mathcal{X}}^n_k}\frac{1}{n}\sum_{i=1}^m\frac{1}{m}\mathbb{E}\left[\frac{\lvert Y_i-\Phi_i \hat{X}^n\rvert}{\lVert\Phi_i\rVert_q}\right]\\
\textrm{subject to } &\lvert\hat{\mathcal{X}}^n_k\rvert \leq 2^{nR} \textrm{ and } X^n(T^c)=\hat{X}^n(T^c).
\end{align*}
The quantity defined above represents the distortion achieved in $Y^m$ corresponding to a particular distortion metric. Since, we force the quantizer to search over quantized versions of the form $\hat{Y}^m=\Phi \hat{X}^n$, the optimization is again carried out over the reconstruction spaces of the form $\hat{\mathcal{X}}^n_k$. Note that we depart from the usual convention of denoting the compression rate as the argument of the distortion rate function in the definition of $\Delta_2(R)$. The argument $R$ in $\Delta_2(R)$ denotes the fact that $2^{nR}$ indices are used for quantization and the compression rate is in fact $\log(2^{nR})/m$.  

Let $D_2(R)$ be the average distortion achieved in $X^n$ by the scheme consisting of compressive sensing followed by quantization and reconstruction. The chain shown in Fig. 1 uses $2^{nR}$ codewords at a compression rate $R$.  In the next section, we present our main result relating $D_2(R)$ and $D_1(R)$. 


\section{Main Results and Analysis}\label{sec:mresult}
We first briefly discuss the order wise optimality of compressed sensing for lossless compression before turning to the main results of this paper. Let us assume for this discussion alone that the non-zero entries of $X^n$ are discrete random variables belonging to some alphabet $\mathcal{X}$ with finite cardinality. Let $T$ denote the sparsity pattern of $X^n$, uniformly distributed among $\binom{n}{k}$ possiblities. Mathematically, $T=\{i:X_i\neq 0\}$. If $R$ is the rate of compression, we have,

\begin{align*}
nR &\geq H(\hat{X}^n)=I(X^n;\hat{X}^n)= H(X^n)-H(X^n|\hat{X}^n)\\
&\overset{(a)}{\geq}H(T,X_i,i\in T)-n\epsilon_n\\
& =H(T)+H(X_i,i\in T|T)-n\epsilon_n\\
&= k\log\frac{n}{k}+k\log e+k\log|\mathcal{X}|-n\epsilon_n.
\end{align*}

Here $(a)$ follows from Fano's inequality where $\epsilon_n\rightarrow 0$ as $n\rightarrow\infty$. Therefore compressed sensing is an order wise optimal lossless compression scheme. 
 
We now turn to the main results of the paper. We state the restricted isometry property (RIP) for matrices \cite{Candes2005}. A matrix $\Phi$, is said to satisfy the $(\epsilon,k)$-RIP  if $\forall x^n\in\mathcal{X}^n_k$,
\begin{equation}
\label{eqn:rip}
(1 - \epsilon) \lVert x^n \rVert_2 \le \lVert\Phi x^n\rVert_2 \le (1 + \epsilon) \lVert x^n \rVert_2.
\end{equation}

We now state a modified version of the RIP which is useful in proving the rate distortion result of the paper when the distortion constraint is on any $p$-norm on the error between the source and the reconstruction. A matrix $\Phi$, is said to satisfy the modified $(\epsilon,k)$-RIP  if $\forall x^n\in\mathcal{X}^n_k$, and $p\geq 1$, 
\begin{equation}
\label{eqn:mrip}
(1-\epsilon)\lVert x^n\rVert_p\leq \sum_{i=1}^m\frac{1}{m}\frac{\lvert\Phi_i x^n\rvert}{\lVert\Phi_i\rVert_q}\leq \lVert x^n\rVert_p(1+\epsilon)
\end{equation}

We show the existence of matrices satisfying the above modified RIP through the following theorem. 
\begin{theorem}\label{thm:mripproof}
Let $\Phi$ be a matrix of dimension $m\times n$ containing entries chosen i.i.d. and supported in $[C_1,C_2]$, where  $-\infty<C_1<C_2<\infty$. For every $\epsilon\in(0,1)$, if $m=\Theta(k\log\frac{n}{k})$, there exists a constant $c_2>0$ such that  with probability greater than $1-2e^{-mc_2}$, 
\begin{equation*}
\sup_{x^n\in \mathcal{X}_k^n}\left\lvert\sum_{i=1}^m\frac{1}{m} \frac{\lvert \Phi_i x^n\rvert}{\lVert x^n\rVert_p\lVert\Phi_i\rVert_q}-1\right\rvert < \epsilon.
\end{equation*}
\end{theorem}

The above theorem is proved in the appendix. We now show that there exist matrices that satisfy both the RIP  and the modified RIP  simultaneously with high probability through the following lemma. Such a result will be useful in proving the main result of the paper. 

\begin{lemma}\label{thm:mrip}
If $m= \Theta(k\log n/k)$, then there exists an $m\times n$ matrix $\Phi$ that satisfies the RIP  (\ref{eqn:rip}) and the modified RIP  (\ref{eqn:mrip}) with high probability.
\end{lemma}
\begin{proof}
Let the entries of $\Phi$ be i.i.d. and distributed according to a Bernoulli distribution taking values $1/\sqrt{m}$ or $-1/\sqrt{m}$ with equal probability. Then, it follows from \cite{Candes2005} that $\Phi$ satisfies the RIP  (\ref{eqn:rip}). Further, using Theorem \ref{thm:mripproof}, it follows that $\Phi$ also satisfies the modified RIP  (\ref{eqn:mrip}) with high probability since the entries of $\Phi$ are i.i.d. and belong to $[-1,1]$. Therefore, $\Phi$ satisfies both (\ref{eqn:rip}) and (\ref{eqn:mrip}) with high probability. 
\end{proof}

The following theorem, which is the main result of the paper states that the coding architecture of Fig. \ref{fig:coding} achieves the same distortion rate function as the optimal compression scheme for the $k$-sparse source.

\begin{theorem}
The coding architecture in Fig. 1 is distortion rate optimal when $\Phi$ satisfies the RIP  (\ref{eqn:rip}) and the modified RIP  (\ref{eqn:mrip}). Mathematically, 
$\forall \epsilon\in(0,1)$, with high probability,
\begin{equation*}
D_2(R)\leq \frac{1+\epsilon}{1-\epsilon}D_1(R).
\end{equation*}
\end{theorem}

\begin{proof}
By Lemma \ref{thm:mrip}, there exists a $\Phi$ that satisfies the RIP  (\ref{eqn:rip}) and the modified RIP  (\ref{eqn:mrip}) simultaneously. A candidate code for the quantization of $Y^m$ can be described as follows. The optimal codebook for $X^n$ is multiplied by $\Phi$ to obtain a codebook for $Y^m$. Now, given a $Y^m$, the quantizer looks for that $\Phi\hat{X}^n$ that minimizes the average distortion. Since $\Phi$ satisfies the modified RIP  with high probability, we have
\begin{align}
(1-\epsilon)\lVert X^n-\hat{X}^n\rVert_p&\leq\sum_{i=1}^m\frac{1}{m}\frac{\lvert Y_i-\Phi_i \hat{X}^n\rvert}{\lVert\Phi_i\rVert_q}\nonumber\\
\Rightarrow (1-\epsilon)\mathbb{E}\left[\lVert X^n-\hat{X}^n\rVert_p\right] &\leq \sum_{i=1}^m\mathbb{E}\left[\frac{1}{m}\frac{\lvert Y_i-\Phi_i \hat{X}^n\rvert}{\lVert\Phi_i\rVert_q}\right]\nonumber\\
\Rightarrow (1-\epsilon)D_2(R)&\leq \Delta_2(R)\label{eqn:step1}
\end{align}

Now, $\forall \delta>0$, let $\tilde{X}^n$ be the quantized value of $X^n$ according to the optimal quantizer that achieves $D_1(R)+\delta$. $\Phi\tilde{X}^n$ is a feasible solution to the problem $\Delta_2(R)$. Therefore, again by the modified RIP, with high probability, 
\begin{equation}\label{eqn:step2}
\Delta_2(R)\leq \frac{1}{n}\sum_{i=1}^m\frac{1}{m}\mathbb{E}\left[\frac{\lvert Y_i-\Phi_i \tilde{X}^n\rvert}{\lVert\Phi_i\rVert_q}\right]\leq \frac{1}{n}(1+\epsilon)\lVert X^n-\tilde{X}^n\rVert_p\leq(1+\epsilon)(D_1(R)+\delta).
\end{equation}
Now, since $\Delta_2(R)\leq (1+\epsilon)(D_1(R)+\delta)$ for all $\delta>0$, we get 
\begin{equation*}
\Delta_2(R)\leq (1+\epsilon)D_1(R).
\end{equation*}
From Equation (\ref{eqn:step1}) and (\ref{eqn:step2}), we get 
\begin{equation*}
D_2(R)\leq \frac{1+\epsilon}{1-\epsilon}D_1(R).
\end{equation*}
We claim that the scheme in Fig. 1 achieves the optimal distortion rate function since $D_2(R)\leq\frac{1+\epsilon}{1-\epsilon}D_1(R)$  for all $\epsilon\in(0,1)$. However, note that as $\epsilon\rightarrow 0$, we require more and more number of measurements to satisfy the RIP and modified RIP  with high probability. Also, forcing $\hat{Y}^m=\Phi \hat{X}^n$ where $\hat{X}^n$ is sparse implies that recovery of the sparse vector is possible without any loss. In other words, $\hat{X}^n$ may be exactly recovered from $\hat{Y}^m$ since $\Phi$ satisfies the RIP  as well. 
\end{proof}

For the specific case of 2-norm distortion measures, the above theorem can be proved for matrices $\Phi$ that just satisfy the RIP alone.  Define 
\begin{align*}
D_1^2(R) = &\inf_{\hat{\mathcal{X}}^n_k}\frac{1}{n}\mathbb{E}\left[\lVert X^n-\hat{X}^n\rVert_2\right] \\
\textrm{subject to }& \lvert\hat{\mathcal{X}}^n_k\rvert \leq 2^{nR} \textrm{ and } X^n(T^c)=\hat{X}^n(T^c)
\end{align*}
and
\begin{align*}
\Delta_2^2(R) = &\inf_{\hat{\mathcal{X}}^n_k}\frac{1}{n}\mathbb{E}\left[\lVert Y^n-\Phi\hat{X}^n\rVert_2\right]\\
\textrm{subject to }& \lvert\hat{\mathcal{X}}_k^n\rvert \leq 2^{nR} \textrm{ and } X^n(T^c)=\hat{X}^n(T^c).
\end{align*}
Let $D_2^2(R)$ be the distortion rate function achieved by the coding architecture of Fig. 1. 

\begin{remark}
The coding architecture of Fig. 1 is distortion rate optimal when $\Phi$ satisfies the RIP  (\ref{eqn:rip}). Mathematically, $\forall \epsilon\in(0,1)$, with high probability,
\begin{equation*}
D_2^2(R)\leq \frac{1+\epsilon}{1-\epsilon}D_1^2(R).
\end{equation*}
\end{remark}
The proof is similar to that of Theorem 2, where only the RIP is used instead of the modified RIP in steps (\ref{eqn:step1}) and (\ref{eqn:step2}).  



\section{Conclusion}\label{sec:conc}
We consider the problem of quantization of sparse signals using compressive sensing. We show that the chain comprising of compressive sampling followed by quantization and then reconstruction is rate distortion optimal when the reconstruction is also required to be sparse. The result is shown when the distortion metric is any $p$-norm, $p\geq 1$, on the error between the source and the reconstruction. The proof of the result requires the compressive sensing matrix to satisfy a new modified restricted isometry property and we also prove the existence of matrices satisfying this property. 


\appendix
Theorem \ref{thm:mrip} is proved through the sequence of the following lemmas. The overall procedure closely follows the technique in \cite{Baraniuk2008} with suitable changes as required. 
We first state a lemma about the concentration of measure around the mean for bounded random variables. Let $Z=|Y_1|+|Y_2|+\ldots+|Y_m|$, where $Y_i$, $i=1,2,\ldots,m$, are independent bounded continuous random variables such that $|Y_i|\leq C$ for each $i$. Also, let $\mathbb{E}[Z]=\beta m$. 
\begin{lemma}\label{thm:concentration}
For $\epsilon\in(0,1)$, the random variable $Z$ satisfies
\begin{equation*}
\mathbb{P}\left[(1-\epsilon)\beta m\leq Z\leq (1+\epsilon)\beta m\right] > 1-2e^{-m\gamma(\epsilon)},
\end{equation*}
with $\gamma(\epsilon)>0$. 
\end{lemma}
\begin{proof}
Following the procedure in \cite{Baraniuk2008}, we use the inequality $e^{-a}\leq 1-a+\frac{a^2}{2}$ for all $a\geq 0$ to get
\begin{equation*}
\mathbb{E}\left[e^{-\lambda |Y_i|}\right]\leq 1-\lambda\mathbb{E}[|Y_i|]+\lambda^2\mathbb{E}[|Y_i|^{2}],
\end{equation*}
for $\lambda>0$. Further, since $1-a\leq e^{-a}$ for all $a\in\mathbb{R}$, we have 
\begin{equation*}
\mathbb{E}\left[e^{-\lambda |Y_i|}\right]\leq  e^{-(\lambda\mathbb{E}[|Y_i|]-\lambda^2\mathbb{E}[|Y_i|^{2}])}.
\end{equation*}
Therefore, by Markov's inequality, we obtain,
\begin{align*}
\mathbb{P}[-Z\geq-(1-\epsilon)\beta m]&=\mathbb{P}\left[e^{-\lambda Z}\geq e^{-\lambda(1-\epsilon)\beta m}\right]\\
&\leq e^{\lambda(1-\epsilon)\beta m}\mathbb{E}\left[e^{-\lambda Z}\right]\\
&=e^{\lambda(1-\epsilon)\beta m}\mathbb{E}\left[e^{-\sum_{i=1}^m\lambda|Y_i|}\right]\\
&\leq e^{-\lambda\left[\sum_{i=1}^m\mathbb{E}[|Y_i|]-m(1-\epsilon)\beta -\lambda\sum_{i=1}^m\mathbb{E}[|Y_i|^{2}]\right]}\\
&\leq e^{-m\lambda\left[\epsilon\beta -\lambda\sum_{i=1}^m\mathbb{E}[|Y_i|^{2}]/m\right]}. 
\end{align*}

For the other side of the inequality, we use Hoeffding's inequality as follows. Now,
\begin{align*}
\mathbb{P}[Z\geq (1+\epsilon)\beta m] &= \mathbb{P}\left[e^{\lambda Z}\geq e^{\lambda(1+\epsilon)\beta m}\right]\\
&\leq e^{-\lambda(1+\epsilon)\beta m}\mathbb{E}\left[e^{\lambda Z}\right]\\
&\leq e^{-\lambda(1+\epsilon)\beta m}\left[e^{\lambda^2C^2/8}\right]^m\\
&= e^{-m\lambda\left[(1+\epsilon)\beta-\lambda C^2/8\right]},
\end{align*}
where we use Markov's inequality in the second step and Hoeffding's inequality in the third step. Choosing $\lambda<\min\left\{\frac{m\epsilon\beta}{\sum_{i=1}^m\mathbb{E}[|Y_i|^{2}]},\frac{(1+\epsilon)8\beta}{C^2}\right\}$, we get 
\begin{equation*}
\mathbb{P}\left[(1-\epsilon)\beta m\leq Z\leq (1-\epsilon)\beta m\right] > 1-2e^{-m\gamma(\epsilon)},
\end{equation*}
with $\gamma(\epsilon)=\lambda\min\Big\{\epsilon\beta-\lambda\frac{\sum_{i=1}^m\mathbb{E}[|Y_i|^2]}{m},(1+\epsilon)\beta-\lambda \frac{C^2}{8}\Big\}>0$. 
\end{proof}

The following lemma states that for every given $x^n$, there exists a matrix that satisfies the condition in the modified isometry property with high probability. Note that this does not prove the statement of the theorem yet since we need to show the existence of a matrix that satisfies the condition for all $x^n\in\mathcal{X}_k^n$.

\begin{lemma}\label{thm:probgivenx}
For every given $x^n$, and $\epsilon\in(0,1)$, an $m\times n$ matrix $\Phi$ with i.i.d. entries supported in $[C_1,C_2]$, where $-\infty<C_1<C_2<\infty$, satisfies
\begin{equation*}
\mathbb{P}\left[ \left\lvert\sum_{i=1}^m\frac{1}{m} \frac{\lvert \Phi_i x^n\rvert}{\lVert x^n\rVert_p\lVert\Phi_i\rVert_q}-1\right\rvert > \epsilon\right] < 2e^{-m\gamma(\epsilon)},
\end{equation*}
where $\gamma(\epsilon)>0$ and $\frac{1}{p}+\frac{1}{q}=1$, $p,q\geq 1$. 
\end{lemma}

\begin{proof}
Since each entry of $\Phi$ is i.i.d. and bounded in $[C_1,C_2]$, by H\"{o}lder's inequality, we have 
\begin{equation*}
\frac{|(\Phi_i x^n)|}{\lVert x^n\rVert_p\lVert \Phi_i\rVert_q}\leq \frac{\lVert x^n\rVert_p\lVert \Phi_i\rVert_q}{\lVert x^n\rVert_p\lVert \Phi_i\rVert_q}\leq 1,
\end{equation*}
where $q$ satisfies $\frac{1}{p}+\frac{1}{q}=1$. Therefore, the random variable $Y_i=\frac{|(\Phi_i x^n)|}{\lVert x^n\rVert_p\lVert \Phi_i\rVert_q}$ satisfies $|Y_i|\leq 1$,
for $i=1,2,\ldots,m$. Let $\sum_{i=1}^m\mathbb{E}[|Y_i|]=m\beta$. It follows that $\beta\leq 1$, since 
\begin{equation*}
\mathbb{E}\left[\frac{|(\Phi_i x^n)|}{\lVert x^n\rVert_p\lVert \Phi_i\rVert_q}\right] \leq \mathbb{E}\left[\frac{\lVert x^n\rVert_p\lVert \Phi_i\rVert_q}{\lVert x^n\rVert_p\lVert \Phi_i\rVert_q}\right]= 1.
\end{equation*}
By applying Lemma \ref{thm:concentration}, we get
\begin{align*}
\mathbb{P}\left[(1-\epsilon)\beta \leq  \sum_{i=1}^m\frac{1}{m}|Y_i|\leq (1+\epsilon)\beta \right] &> 1-2e^{-m\gamma(\epsilon)}\\
\Rightarrow \mathbb{P}\left[(1-\epsilon)\leq \sum_{i=1}^m\frac{1}{m}|Y_i|\leq (1+\epsilon)\right] &> 1-2e^{-m\gamma(\epsilon)},
\end{align*} 
since $\beta\leq 1$. The desired result follows from the above.
\end{proof}

We now present a lemma about the quantization of vectors in the unit $p$-norm ball. We characterize the size of the set required to represent every such vector within a prescribed $p$-norm error. 

\begin{lemma}\label{thm:quant}
For $\epsilon\in(0,1)$, there exists a finite set $\mathcal{Q}\subset\mathbb{R}^n$ such that $|\mathcal{Q}|\leq \sqrt{\frac{n}{2\pi p}}\left(\frac{c_1}{\epsilon}\right)^n$ and 
\begin{equation*}
\sup_{x^n:\lVert x^n\rVert_p\leq 1}\min_{v^n\in\mathcal{Q}}\lVert x^n-v^n\rVert_p \leq \epsilon.
\end{equation*}
\end{lemma}

\begin{proof}
For $k\in\mathbb{N}$, a natural number, define
\begin{equation*}
\mathcal{Q}^{'} = \{ x^n: x_i=\frac{j}{k} \textrm{ for some } j\in\{-k,-k+1,\ldots,k\}\}.
\end{equation*}

$\mathcal{Q}'$ is a set of quantization indices in $n$ dimensions with size $(2k+1)^n$. We now define $\mathcal{Q}=\mathcal{Q}'\cap \mathcal{B}_p(1)$, where $\mathcal{B}_p(1)$ is the unit ball in $\mathbb{R}^n$ with $L_p$ norm as the distance metric. The size of $\mathcal{Q}$, is then the ratio of the volumes of the unit ball $\mathcal{B}_p(1)$ to the unit cube times the size of $\mathcal{Q}'$. The volume of $\mathcal{B}_p(1)$ is given by
\begin{equation*}
\vol(\mathcal{B}_p(1)) = \frac{\Gamma\left(\frac{p+1}{p}\right)^n}{\Gamma\left(\frac{p+n}{p}\right)}.
\end{equation*}
Therefore,
\begin{equation*}
|\mathcal{Q}| \leq (2k+1)^n\frac{\Gamma\left(\frac{p+1}{p}\right)^n}{\Gamma\left(\frac{p+n}{p}\right)}.
\end{equation*}

We now specify the choice of $k$ that bounds $|\mathcal{Q}|$ as required. Let $v^n$ be the quantization index for $x^n$. Mathematically, for $i=1,2,\ldots,n$, we choose, $v_i=\sign(x_i) \frac{\lfloor|x_i|k\rfloor}{k}$. Therefore, $|x_i-v_i|\leq 1/k$ and 
\begin{equation*}
\lVert x^n-v^n\rVert_p \leq \frac{n^{1/p}}{k}. 
\end{equation*}
Choosing $k=\lceil n^{1/p}\rceil/\epsilon$, we satisfy $\lVert x^n-v^n\rVert_p \leq \epsilon$, and obtain 
\begin{equation*}
|\mathcal{Q}|\leq (3\lceil n^{1/p}\rceil/\epsilon)^n\frac{\Gamma\left(\frac{p+1}{p}\right)^n}{\Gamma\left(\frac{p+n}{p}\right)}.
\end{equation*}
Now,
\begin{equation*}
\Gamma\left(1+n/p\right) \geq \sqrt{\frac{2\pi p}{n}} \left(\frac{n}{ep}\right)^{n/p},
\end{equation*}

and $\Gamma\left(\frac{p+1}{p}\right)\leq 1$. Therefore 
\begin{equation*}
|\mathcal{Q}|\leq \frac{6^n n^{n/p}}{\epsilon^n} \sqrt{\frac{n}{2\pi p}}\left(\frac{ep}{n}\right)^{n/p}<\sqrt{\frac{n}{2\pi p}}\left(\frac{c_1}{\epsilon}\right)^n,
\end{equation*}
where $c_1=6(ep)^{1/p}$.
\end{proof}

We now prove Theorem 1. $\Phi$ is a matrix of dimension $m\times n$ containing entries chosen i.i.d. and supported in $[C_1,C_2]$, where $-\infty<C_1<C_2<\infty$. We need to show that for every $\epsilon\in(0,1)$, if $m=\Theta(k\log\frac{n}{k})$, there exists a constant $c_2>0$  such that with probability greater than $1-2e^{-mc_2}$, we have 
\begin{equation*}
\sup_{x^n\in \mathcal{X}_k^n}\left\lvert\sum_{i=1}^m\frac{1}{m} \frac{\lvert \Phi_i x^n\rvert}{\lVert x^n\rVert_p\lVert\Phi_i\rVert_q}-1\right\rvert < \epsilon. 
\end{equation*}

Without loss of generality, we consider $x^n$ such that $\lVert x^n\rVert_p=1$. We first prove that for the set of $k$-sparse vectors $x^n$ with a given sparsity pattern, $\Phi$ exists with probability greater than $1-2(4c_1/\epsilon)^ke^{-m\gamma}$. 

By Lemma \ref{thm:quant}, there exists a set $\mathcal{Q}$ with size $|\mathcal{Q}|\leq \sqrt{\frac{k}{2\pi p}}(8c_1/\epsilon)^k$, such that 
\begin{equation*}
\sup_{x^n:\lVert x^n\rVert=1}\min_{v^n\in\mathcal{Q}}\lVert x^n-v^n\rVert \leq \epsilon/8.
\end{equation*}

Now, we show that there exists a matrix $\Phi$ such that
\begin{equation}\label{eqn:qrip}
\sup_{v^n\in\mathcal{Q}}\left\lvert\sum_{i=1}^m\frac{1}{m} \frac{\lvert \Phi_i v^n\rvert}{\lVert v^n\rVert_p\lVert\Phi_i\rVert_q}-1\right\rvert < \epsilon/2 
\end{equation}
with high probability. By Lemma \ref{thm:probgivenx} and union bounding argument, we have, 
\begin{equation*}
\mathbb{P}\left[\sup_{v^n\in\mathcal{Q}}\left\lvert\sum_{i=1}^m\frac{1}{m} \frac{\lvert \Phi_i v^n\rvert}{\lVert v^n\rVert_p\lVert\Phi_i\rVert_q}-1\right\rvert>\epsilon/2\right] \leq |\mathcal{Q}|2e^{-m\gamma(\epsilon/2)} \leq 2\sqrt{\frac{k}{2\pi p}}(8c_1/\epsilon)^ke^{-m\gamma(\epsilon/2)}.
\end{equation*}

We now prove the statement of the theorem. By H\"{o}lder's inequality,
\begin{equation*}
\sum_{i=1}^m\frac{1}{m}\frac{\lvert\Phi_i x^n\rvert}{\lVert x^n\rVert_p\lVert \Phi_i\rVert_q}\leq 1.
\end{equation*}
Thus, $\sum_{i=1}^m\frac{1}{m}\frac{\lvert\Phi_i x^n\rvert}{\lVert x^n\rVert_p\lVert \Phi_i\rVert_q}\leq 1+\epsilon$. Now, for the other side, using (\ref{eqn:qrip}),
\begin{equation*}
\sum_{i=1}^m\frac{1}{m}\frac{\lvert\Phi_i x^n\rvert}{\lVert x^n\rVert_p\lVert \Phi_i\rVert_q}\geq \sum_{i=1}^m\frac{1}{m}\frac{\lvert\Phi_i v^n\rvert-\lvert\Phi_i (x^n-v^n)\rvert}{\lVert x^n\rVert_p\lVert \Phi_i\rVert_q}\geq (1-\epsilon/2)(1-\epsilon/8)-\epsilon/8\geq 1-\epsilon.
\end{equation*}

Now, considering all $\binom{n}{k}\leq (en/k)^k$, $k$ sparse vectors, the probability that there does not exist $\Phi$ satisfying the modified RIP  is upper bounded by 
\begin{equation*}
2(en/k)^k\sqrt{\frac{k}{2\pi p}}(8c_1/\epsilon)^ke^{-m\gamma(\epsilon/2)} \leq e^{k[\log(en/k)+\log(8c_1/\epsilon)]+\frac{1}{2}\log\frac{k}{2\pi p}-m\gamma(\epsilon/2)+\log2}.
\end{equation*}
Therefore, if 
\begin{equation*}
m > \frac{1}{\gamma(\epsilon/2)}\left(k[\log(en/k)+\log(8c_1/\epsilon)]+\frac{1}{2}\log\frac{k}{2\pi p}+\log2\right),
\end{equation*}
there exists $c_2>0$, chosen smaller than $\gamma(\epsilon/2)-\frac{1}{m}\left(k[\log(en/k)+\log(8c_1/\epsilon)]+\frac{1}{2}\log\frac{k}{2\pi p}+\log2\right)$, such that probability that there does not exist a matrix satisfying the $p$-norm condition is upper bounded by $2e^{-mc_2}$.

\end{document}